\documentclass[smallextended]{svjour3}       
\smartqed  
\usepackage{graphicx}
\usepackage{color}
\usepackage{amsmath}
\usepackage{amssymb}
 \usepackage{a4wide}
\usepackage{subfigure}
\usepackage{algorithm,algorithmic}

 \newtheorem{ex}[definition]{Example}
 \newtheorem{prop}[definition]{Proposition}
 \newtheorem{lem}[definition]{Lemma}
 \newtheorem{thm}[definition]{Theorem}

 \newtheorem{construction}{Construction}
\newtheorem{rem}[definition]{Remark}

\newcommand{\deff}{\mbox{$\stackrel{\rm def}{=}$}}
\newcommand{\field}[1]{\mathbb{#1}}

\newcommand{\F}{\field{F}}

\newcommand{\cA}{{\mathcal A}}

\newcommand{\cC}{{\mathcal C}}

\newcommand{\Gr}{\mathcal{G}_{q}(k,n)}
\newcommand{\G}{\mathcal{G}_{q}(k,n)}
\newcommand{\Uvs}{\mathcal{U}}
\newcommand{\Vvs}{\mathcal{V}}
\newcommand{\Cvs}{\mathcal{C}}
\newcommand{\Rvs}{\mathcal{R}}
\newcommand{\rs}{\mathrm{rs}}

\begin{document}

\title{List Decoding of Lifted Gabidulin Codes via the Pl\"ucker Embedding}


\author{Anna-Lena Trautmann \and  Natalia Silberstein 
\and Joachim Rosenthal
\thanks{First and third author were partially supported by Swiss National Science Foundation Grant no. 138080. The first author was partially supported by Forschungskredit of the University of Zurich, grant no. 57104103.}        
}


\institute{A.-L. Trautmann \and J. Rosenthal  \at
              Institute of Mathematics, University of Zurich, Switzerland\\
              \email{trautmann@math.uzh.ch, rosenthal@math.uzh.ch}           
           \and
           N. Silberstein \at
               Department of Electrical
and Computer Engineering, University of Texas at Austin,
USA \\
\email{natalys@austin.utexas.edu}.
}


\maketitle

\begin{abstract}
Codes in the Grassmannian have
recently found an application in random network coding. All the codewords
in such  codes are subspaces of $\F_q^n$ with a given dimension.

In this paper, we consider the problem of list decoding of a certain family of codes in the Grassmannian, called lifted Gabidulin codes.
 For this purpose we use the Pl\"ucker embedding of the Grassmannian. We describe a way of representing a subset of the Pl\"ucker coordinates of lifted Gabidulin codes as linear block codes. The union of the parity-check equations of these block codes and the equations which arise from the description of a ball around a subspace in the Pl\"ucker coordinates describe the list of codewords with distance less than a given parameter from the received word.
\end{abstract}

\section{Introduction}

 Let $\F_q$ be a finite field of size $q$.
The\textit{ Grassmannian space} (Grassmannian, in short), denoted
by $\Gr$, is the set of all $k$-dimensional subspaces of the
vector space~\smash{$\F_q^n$}, for any given two integers $k$ and
$n$, $0 \le k \le n$. A   subset  $\cC$ of the Grassmannian  is called an $(n,M,d_S,k)_q$
\emph{constant dimension code}
if it has size $M$ and minimum subspace distance $d_S$, where the distance function
in $\Gr$ is defined as follows:
\begin{equation}
\label{def_subspace_distance}
d_S (\Uvs,\Vvs) =  2k -2 \dim\bigl( \Uvs\, {\cap}\Vvs\bigr),
\end{equation}
for any two subspaces $\Uvs$ and $\Vvs$ in $\Gr$.

These codes gained a lot of interest due to the work by K\"otter and
Kschischang~\cite{KK}, where they
show the application of such codes for error-correction in random network coding. They proved that
an $(n,M,d_S,k)_q$ code can correct any $t$ packet errors (which is equivalent to $t$ packet insertions and $t$ packet deletions)  and
any $\tau$ packet erasures introduced  anywhere in
the network as long as $4t + 2\tau < d_S$. This application has motivated extensive work in the area~\cite{BoGa09,EtSi09,EV08,GaYa10,GNW12,GuXi12,KoKu08,MGR08,MaVa12,MaVa10,ro12b,Ska10,TMR10,TrRo10}. In the same work the before mentioned authors gave a Singleton like upper
bound on the size of such codes and a Reed-Solomon like code which
asymptotically attains this bound. Silva, K\"otter, and
Kschischang~\cite{si08j} showed how this construction can be
described in terms of lifted Gabidulin codes~\cite{ga85a}. The generalizations of this construction and the decoding algorithms were presented in~\cite{BoGa09,EtSi09,MGR08,ro12b,Ska10,TrRo10}. Another type of constructions (orbit codes) can be found in~\cite{EV08,KoKu08,TMR10}.

In this paper we  focus on the list decoding of  lifted Gabidulin codes.
 For the classical Gabidulin codes it was recently shown by Wachter-Zeh~\cite{WZ12} that, if the radius of the ball around a received word is greater than the Johnson radius,  no polynomial-time list decoding is possible, since the list size can be exponential. Algebraic list decoding algorithms for folded Gabidulin codes were discussed in~\cite{GNW12,MaVa12}. The constructions of subcodes of lifted Gabidulin codes and their algebraic list decoding algorithms were presented in~\cite{GuXi12,MaVa10}.

Our approach for list decoding codes in the Grassmannian is to apply the techniques of Schubert calculus over finite fields, i.e. we
 represent subspaces  in the Grassmannian by their Pl\"ucker coordinates.
It was proven in~\cite{ro12b} that a ball of a given radius (with respect to the subspace distance) around a subspace can be described by explicit linear equations in the Pl\"ucker embedding.  In this work we describe a way of representing a subset of the Pl\"ucker coordinates of lifted Gabidulin codes as linear block codes, which results in additional linear (parity-check) equations. The solutions of all these equations will constitute the resulting list of codewords.

The rest of this paper is organized as follows. In Section~\ref{sec:preliminaries} we describe the construction of Gabidulin and lifted Gabidulin codes and discuss the Pl\"ucker embedding of subspaces in the Grassmannian. In Section~\ref{sec:list_decoding} we  describe a representation of a subset of the Pl\"ucker coordinates of a lifted Gabidulin code and present a list decoding algorithm.
Conclusions and problems for future research are given in Section~\ref{sec:conclusions}.


\section{Preliminaries and Notations}
\label{sec:preliminaries}

We denote by $GL_{n}$ the general linear group over $\F_{q}$, by $S_{n}$ the symmetric group of degree $n$.
  With $\mathbb P^{n}$ we denote the projective space of order $n$ over $\F_{q}$.

Let $p(x)=\sum p_i x^i\in \F_q[x]$ be a monic and irreducible polynomial of degree $\ell$, and $\alpha$ be a root of $p(x)$. Then it holds that $\F_{q^{\ell}} \cong \F_q[\alpha]$. We denote the vector space isomorphism between the extension field $\F_{q^{\ell}}$ and the vector space $\F_q^{\ell}$ by
\begin{align*}
 \phi^{(\ell)} : \quad\F_{q^{\ell}} & \longrightarrow \F_q^{\ell} \\
          \sum_{i=0}^{\ell-1} \lambda_i \alpha^i &\longmapsto (\lambda_0,\dots,\lambda_{\ell-1})  .
\end{align*}

Moreover, we need the following notations: $\rs(U)$ denotes the row space of a matrix $U$,
\[\binom{[n]}{k} := \{(x_{1},\dots, x_{k}) \mid x_{i} \in \{1,2,\dots,n\}, x_{1}<\dots<x_{k}\}, \]
and for a matrix $A$ we denote its $i$-th row by $A[i]$, its $i$-th column by $A_i$, and the entry in the $i$-th row and the $j$-th column by $A_{i,j}$.


\subsection{Lifted Gabidulin (LG) Codes}

For two $k \times \ell$ matrices $A$ and $B$ over $\F_q$ the {\it
rank distance} is defined by
$$
d_R (A,B) \deff \text{rank}(A-B)~.
$$
A  $[k \times \ell,\varrho,\delta]$ {\it rank-metric code} $C$
is a linear subspace with dimension $\varrho$
of $\F_q^{k \times \ell}$, in which each two distinct codewords $A$
and $B$ have distance $d_R (A,B) \geq \delta$. For a $[k \times
\ell,\varrho,\delta]$ rank-metric code $C$ it was proven
in~\cite{de78,ga85a,ro91} that
\begin{equation}
\label{eq:MRD} \varrho \leq
\text{min}\{k(\ell-\delta+1),\ell(k-\delta+1)\}~.
\end{equation}
 The codes which attain this bound are called {\it
maximum rank distance} codes (or MRD codes in short).

An important family of MRD linear codes was presented by Gabidulin~\cite{ga85a}.
These codes can be seen as the analogs of Reed-Solomon codes for the rank metric.
From now on let  $k\leq \ell$.
A codeword $A$ in a $[k \times \ell, \varrho , \delta]$
rank-metric code $C$, can be represented by a vector $c_{A}=(c_1 , c_2 , \ldots , c_{k})$, where $c_i={\phi^{(\ell)}}^{-1}(A[i]) \in \F_{q^{\ell}}$.
Let $g_i\in \F_{q^{\ell}}$, $h_i\in \F_{q^{\ell}}$, $1\leq i\leq k$,  be two sets of linearly independent over $\F_q$ elements.
Then the generator matrix $G$ and the parity-check matrix $H$ of a  $[k \times \ell,\varrho,\delta]$
Gabidulin MRD code are given by
$$
G=\left(\begin{array}{cccc}
g_{1} & g_{2} & \ldots & g_{k}\\
g_{1}^{[1]} & g_{2}^{[1]} & \ldots & g_{k}^{[1]}\\
g_{1}^{[2]} & g_{2}^{[2]} & \ldots & g_{k}^{[2]}\\
\vdots & \vdots & \vdots & \vdots\\
g_{1}^{[k-\delta]} & g_{2}^{[k-\delta]} & \ldots & g_{k}^{[k-\delta]}\end{array}\right),
H=\left(\begin{array}{cccc}
h_{1} & h_{2} & \ldots & h_{k}\\
h_{1}^{[1]} & h_{2}^{[1]} & \ldots & h_{k}^{[1]}\\
h_{1}^{[2]} & h_{2}^{[2]} & \ldots & h_{k}^{[2]}\\
\vdots & \vdots & \vdots & \vdots\\
h_{1}^{[\delta-2]} & h_{2}^{[\delta-2]} & \ldots & h_{k}^{[\delta-2]}\end{array}\right),
$$
where  $\varrho=\ell (k-\delta+1)$, and  $[i]=q^{i}$~\cite{ga85a}.

Let $A$ be a $k \times \ell$ matrix over $\F_q$ and let $I_k$ be the
$k \times k$ identity matrix. The matrix $[ I_k ~ A ]$ can be
viewed as a generator matrix of a $k$-dimensional subspace of
$\F_q^{k+\ell}$. This subspace is called the \emph{lifting} of $A$~\cite{si08j}.

When the codewords of a rank-metric code
$C$ are lifted to $k$-dimensional subspaces, the result is
a constant
dimension code $\cC$. If $C$ is a Gabidulin MRD code
then $\cC$ is called a \emph{lifted Gabidulin (LG)
code}~\cite{si08j}.

\begin{thm}\cite{si08j}
\label{trm:param lifted MRD}
Let $k$, $n$ be positive integers such that ${k \leq n-k}$.
If $C$ is a $[k \times (n-k), (n-k)(k-\delta +1),\delta ]$ Gabidulin MRD
code then $\cC$ is an $(n,q^{(n-k)(k-\delta+1)},2\delta, k)_{q}$ constant dimension code.
\end{thm}

%


\subsection{The Pl\"ucker Embedding}\label{sec22}

The basic idea of using the Pl\"ucker embedding for list decoding of subspace codes was already stated in \cite{ro12b,tr12a}. We will now recall the main definitions and theorems from those works. The proofs of the results can also be found in there. For more information or a more general formulation of the Pl\"ucker embedding and its applications the interested reader is referred to \cite{ho52}.

Let $U\in \F_{q}^{k\times n}$ such that its row space $\rs (U)$ describes the subspace
$\Uvs \in \Gr$. $M_{i_1,\dots,i_k}(U)$ denotes the minor of $U$ given by the columns $i_1,\dots,i_k$.
The Grassmannian $\Gr$ can be embedded into projective space
using the Pl\"ucker embedding:
\begin{align*}
\varphi : \Gr  &\longrightarrow \mathbb{P}^{\binom{n}{k}-1} \\
\rs (U) &\longmapsto [M_{1,...,k}(U) : M_{1,...,k-1,k+1}(U) :\ldots : M_{n-k+1,...,n}(U)].
\end{align*}
%
%
The $k\times k$ minors $M_{i_1,\ldots ,i_k}(U)$ of the matrix $U$
are called the \emph{Pl\"ucker coordinates} of the subspace $\Uvs$. By convention, we order the minors lexicographically by the column indices.

The image of this embedding describes indeed a variety and the
defining equations of the image are given by the so called \emph{shuffle relations}
(see e.g.~\cite{kl72,pr82}), which are
multilinear equations of monomial degree $2$ in terms of the Pl\"ucker coordinates:
\begin{prop}
Consider $x:=[x_{1,\dots,k}:\dots:x_{n-k+1,\dots,n}]\in \mathbb{P}^{\binom{n}{k}-1}$. Then there exists a $\Uvs \in \Gr$ such that $\varphi(\Uvs)=x$ if and only if
 \[\sum_{\sigma\in S_{2k}} {\mathrm{sgn}({\sigma})} x_{\sigma(i_{1},\dots,i_{k})} x_{\sigma(i_{k+1},\dots,i_{2k})}   = 0 \quad \forall (i_1,\dots,i_{2k}) \in \binom{[n]}{2k}.\]
\end{prop}

Then one can easily count the number of different shuffle equations.

\begin{lem}\label{lem:num}
 There are $\binom{n}{2k}$ shuffle relations defining $\Gr$ in the Pl\"ucker embedding.
\end{lem}

\begin{ex}\label{ex:shuffle}
$\mathcal{G}_q(2,4)$ is  described
by a single relation: $$x_{12}x_{34}-x_{13}x_{24}+x_{14}x_{23} = 0.$$
\end{ex}

The balls of radius $2t$ (with respect to the subspace
distance) around some $\Uvs \in \Gr $ can be described by explicit equations in the Pl\"ucker embedding.
 For it we need the \emph{Bruhat order}:
\[(i_{1},...,i_k) \geq (j_{1},...,j_{k}) \iff
 i_{l} \geq j_{l}  \quad\forall l \in \{1,\dots,k\} .\]
Note, that the Bruhat order is not a total but only a partial order on $\binom{[n]}{k}$.

\begin{ex}According to the Bruhat order it holds that $(1,2,7)\leq(2,3,7)$. But the fact that $(2,4,6)\nleq(2,3,7)$ does not imply that $(2,4,6)>(2,3,7)$. These two tuples are not comparable.
\end{ex}

The equations defining the balls are easily determined in the following special case:

\begin{prop}\cite{ho52,ro12b}\label{prop5}
  Define $\Uvs_{0}:=\rs [\begin{array}{cc}I_{k} &0_{k\times
      n-k} \end{array}]$. Then
\begin{align*}
B_{2t}(\Uvs_{0}) = \{\Vvs=\rs (V) \in & \Gr  \mid M_{i_1,...,i_{k}}(V) = 0  \\ &\forall (i_{1},...,i_{k})
\not \leq  (t+1,\dots,k,n-t+1,...,n)  \}.
\end{align*}
\end{prop}


 With the knowledge of $B_{2t}(\Uvs_0)$ we can also express
 $B_{2t}(\Uvs)$ for any $\Uvs \in \Gr $. For this note, that for any
 $\Uvs \in \Gr $ there exists an $A\in GL_n$ such that $\Uvs_0 A= \Uvs$.
 Moreover,
\[B_{2t}(\Uvs_0 A) = B_{2t}(\Uvs_0) A .\]

The following results are taken from \cite{ro12b}, where also the respective proofs can be found.

For simplifying the computations we define $\bar{\varphi}$ 
on $GL_n$, where we denote by $A[{i_{1},\dots, i_k}]$ the submatrix of $A$ that consists of the rows $i_{1}, \dots, i_{k}$:
\begin{align*}
\bar{\varphi} : GL_n &\longrightarrow GL_{\binom{n}{k}} \\
 A & \longmapsto \left(\begin{array}{cccccc}
\det A_{1,\dots, k}[1, \dots, k] & \dots & \det A_{n-k+1 ,\dots, n}[{1,\dots, k}]\\
\vdots & & \vdots \\
\det A_{1, \dots, k}[{n-k+1,\dots, n}] & \dots & \det A_{n-k+1,\dots, n}[n-k+1 ,\dots, n]
                   \end{array}
 \right)
\end{align*}


\begin{lem}\label{lem:2}
Let $\Uvs \in \Gr $ and $A\in GL_{n}$. It holds that
\[\varphi(\Uvs A) = \varphi(\Uvs) \bar{\varphi}(A).\]
\end{lem}

\begin{thm}\label{thm5}
Let $\Uvs= \Uvs_{0}A \in \Gr $. Then
\begin{align*}
B_{2t}(\Uvs) =B_{2t}(\Uvs_{0} A)
=\{\Vvs  \in \Gr  \mid M_{i_{1},\dots,i_{k}}(V)\bar \varphi(A^{-1}) = 0
\\\forall (i_{1},\dots,i_{k}) \not \leq  (t+1,\dots,k,n-t+1,...,n)\}.
\end{align*}
\end{thm}

There are always several choices for $A\in GL_{n}$ such that
$\Uvs_{0}A=\Uvs$. Since $GL_{{\binom{n}{k}}}$ is very large we try to
choose $A$ as simple as possible. We will now explain one such
construction.
\begin{construction}\label{con4}
For a given $\Uvs=\rs (U) \in \Gr$ we construct $A\in GL_{n}$ such that $\Uvs_{0}A=\Uvs$ as follows:
\begin{enumerate}
\item The first $k$ rows of $A$ are equal to the matrix representation
  $U$ of $\Uvs$.
\item Find the pivot columns of $U$ (assume that $U$ is in RREF).
\item Fill up the respective columns of $A$ with zeros in the lower $n-k$ rows.
\item Fill up the remaining submatrix of size $n-k\times n-k$ with an identity matrix.
\end{enumerate}
Then the inverse of $A$ can be computed as follows:
\begin{enumerate}
\item Find a permutation $\sigma \in S_{n}$ that permutes the columns
  of $A$ such that
\[\sigma(A)= \left(\begin{array}{ccc} I_k & U'' \\ 0 & I_{n-k} \end{array}\right) .\]
\item Then the inverse of that matrix is
\[\sigma(A)^{-1}= \left(\begin{array}{ccc} I_k & -U'' \\ 0 & I_{n-k} \end{array}\right).\]
\item Apply $\sigma$ on the rows of $\sigma(A)^{-1}$. The result is
  $A^{-1}$.  One can easily see this if one represents $\sigma$ by a
  matrix $S$. Then one gets $(SA)^{-1}S=A^{-1}S^{-1}S=A^{-1}$.
\end{enumerate}
\end{construction}

Thus, we know how to describe the balls of a given radius $2t$ around an element of $\Gr$ with linear equations in the Pl\"ucker embedding, which is exactly what is needed for a list decoding algorithm. In the following section we will describe a way of representing a subset of the Pl\"ucker coordinates of lifted rank-metric codes as linear block codes, which can then be used to come up with a list decoding algorithm in the Pl\"ucker embedding.


\section{List Decoding LG Codes in the Pl\"ucker Embedding}
\label{sec:list_decoding}

\subsection{Linear Block Codes over $\F_q$ in the Pl\"ucker Coordinates of LG Codes}\label{sec:21}

Let $C$ be an $[k\times(n-k), (n-k)(k-\delta+1,\delta)]$ Gabidulin MRD code over $\F_q$. Then  by Theorem~\ref{trm:param lifted MRD} its lifting is a code $\cC$ of size $q^{(n-k)(k-\delta+1)}$ in the Grassmannian $\Gr$. Let
$$x^{\cA}=[x^{\cA}_{1\ldots k}:\ldots : x^{\cA}_{n-k+1\ldots n}] \in \mathbb{P}^{\binom{n}{k}-1}$$
be a vector which represents the Pl\"ucker coordinates of a subspace $\cA\in \Gr$. If $x^{\cA}$ is normalized (i.e. the first non-zero entry is equal to one), then $x^{\cA}_{1\ldots k}=1$ for any $\cA\in \cC$.

Let $[k]=\{1,2,\ldots,k\}$, and let $\underline{i}=\{i_1,i_2,\ldots,i_k\}$  be a set of indices such that $|\underline{i}\cap[k]|=k-1$. Let $t\in \underline{i}$, such that $t>k$, and $s=[k]\setminus\underline{i}$.

\begin{lem}\label{lem7}
Consider $A\in C$ and $\cA=\rs[\;I_k \; A\;]$. If $x^{\cA}$ is normalized,
then $x^{\cA}_{\underline{i}}=(-1)^{k-s}A_{s,t-k}$.
\end{lem}
\begin{proof}
It holds that $x^{\cA}$ is normalized if its entries are the minors of the reduced row echelon form of $\cA$, which is $[\;I_k \; A\;]$. Because of the identity matrix in the first $k$ columns, the statement follows directly from the definition of the Pl\"ucker coordinates.
\qed \end{proof}

Note, that we have to worry about the normalization since $x^{\cA}$ is projective. In the following we will always assume that any element from $\mathbb{P}^{\binom{n}{k}-1}$ is normalized.

With Lemma \ref{lem7} one can easily show, that a subset of the Pl\"ucker coordinates of a lifted Gabidulin code form a linear code over $\F_q$:

\begin{thm} \label{thm9}
The restriction of the set of Pl\"ucker coordinates of an $(n,q^{(n-k)(k-\delta+1)},2\delta, k)_{q}$ lifted Gabidulin code $\cC$ to the set $\{\underline{i}:|\underline{i}|=k, |\underline{i}\cap[k]|=k-1\}$ forms a linear code $C^p$ over $\F_q$ of length $k(n-k)$, dimension $(n-k)(k-\delta+1)$ and minimum distance $d_{min}\geq \delta$.
\end{thm}

\begin{proof}
 Since $C$ is linear, it holds that for every $A,B\in C$ we have $A+B\in C$. Together with Lemma~\ref{lem7} we have the same property when we consider the restriction of the set of Pl\"ucker coordinates of a lifted Gabidulin code to the set $\{\underline{i}:|\underline{i}|=k, |\underline{i}\cap[k]|=k-1\}$. This set is of size $k(n-k)$, and therefore we obtain a linear code $C^p$ of length $k(n-k)$ and the same dimension as~$C$, i.e. $(n-k)(k-\delta+1)$. Since the rank of each non-zero $A\in C$ is greater or equal to $\delta$, also the number of non-zero entries of $A$ has to be greater or equal to $\delta$, hence the minimum Hamming distance $d_{min}$ of $C^p$ satisfies $d_{min}\geq \delta$.
\qed \end{proof}

\begin{ex}\label{ex8}
Let $\alpha \in \F_{2^2}$ be a primitive element,  fulfilling $\alpha^2=\alpha+1$.
Let $C$ be a $[2\times2, 2, \delta=2]$ Gabidulin MRD code over $\F_2$ with parity-check and generator matrices given by
\[H=(1\; \alpha)\textmd{ and }G=(\alpha\;1),
\]
respectively.
Hence, we want to lift $C=\{(b\alpha,b): b\in \F_{2^2}\}$.
The codewords of $C$, their representation as $2\times2$ matrices, their lifting  to $\mathcal{G}_2(2,4)$ and the respective Pl\"ucker coordinates are given in the following table.

\begin{center}
\begin{tabular}{|c|c|c|c|}
  \hline
  vector representation  & matrix representation & lifting & Pl\"ucker coordinates \\
  \hline
  \hline
$(0,0)$&
      $ \left( \begin{array}{cc}
                 0 & 0 \\
                 0 & 0 \\
               \end{array}\right)$ &
                $\left(  \begin{array}{cccc}
          1 & 0 & 0 & 0 \\
          0 & 1 & 0 & 0 \\  \end{array} \right) $ &
           $[1:0:0:0:0:0]$ \\\hline
$(\alpha,1)$ &
       $\left( \begin{array}{cc}
                 0 & 1\\
                 1 & 0 \\
               \end{array} \right)$ &
                $\left(  \begin{array}{cccc}
          1 & 0 & 0 & 1 \\
          0 & 1 & 1 & 0 \\\end{array}\right)$ &
           $[1:1:0:0:1:1]$ \\  \hline
$(\alpha^2,\alpha)$&
       $\left(      \begin{array}{cc}
                 1 & 1 \\
                 0 & 1 \\ \end{array} \right)$ &
                   $\left( \begin{array}{cccc}
          1 & 0 & 1 & 1 \\
          0 & 1 & 0 & 1 \\\end{array}    \right) $ &
           $[1:0:1:1:1:1]$ \\\hline
$(1,\alpha^2)$&
      $ \left(\begin{array}{cc}
                 1 & 0 \\
                 1 & 1 \\
               \end{array} \right) $&
                $\left(\begin{array}{cccc}
          1 & 0 & 1 & 0 \\
          0 & 1 & 1 & 1 \\
        \end{array}\right)$ & $[1:1:1:1:0:1]$ \\\hline
\end{tabular}
\end{center}
%
%
In this example, $C^p=\{(0000),(1001),(0111),(1110)\}$. This is a $[4,2,2]$  linear code in the Hamming space. Its parity-check matrix is
$$H^p=\left(
        \begin{array}{cccc}
          1 & 0 & 1 & 1 \\
          0 & 1 & 1 & 0 \\
        \end{array}
      \right).
$$
In other words,  a Pl\"ucker coordinate vector $[x_{12}:x_{13}:x_{14}:x_{23}:x_{24}:x_{34}]$ of a vector space from $\mathcal{G}_{2}(2,4)$ represents a codeword of the lifted Gabidulin code from above if and only if $x_{12}=1$, $x_{14}+x_{23}=0$, and $x_{13}+x_{23}+x_{24}=0$.

\end{ex}

\subsection{The List Decoding Algorithm}

We now have all the machinery needed to describe a list decoding algorithm for lifted rank-metric codes in the Pl\"ucker coordinates under the assumption that the received word has the same dimension as the codewords.
Consider a lifted rank-metric code $\mathcal{C} \subseteq \G$ and denote its corresponding $[k(n-k),(n-k)(k-\delta+1)]$-linear block code over $\F_q$ 
by $C^{p}$. The corresponding parity check matrix is denoted by $H^{p}$. Let $\mathcal{R} =\rs(R) \in \Gr$ be the received word. Let $e$ be the number of errors (i.e. insertions and deletions) to be corrected.

We showed in Section \ref{sec:21} how a subset of the Pl\"ucker coordinates of a LG code forms a linear block code that is defined through the parity check matrix $H^{p}$. Since we want to describe a list decoding algorithm inside the whole set of Pl\"ucker coordinates, we define an extension of $H^{p}$ as follows:
\[\bar H^{p} = \left(\begin{array}{cccc}
0_{(\delta-1)(n-k)\times 1} & H^{p} &  0_{(\delta-1)(n-k)\times \ell}
\end{array}\right)\]
where $\ell = \binom{n}{k} - k(n-k) -1$. Then $[x_{1\dots k}:\ldots : x_{n-k+1\dots n}]\bar {H^{p}}^{T} = 0$ gives rise to the same equations as $[x_{i_{1}}:\ldots:x_{i_{k(n-k)}}]{H^{p}}^{T}=0$, for $i_{1},\dots,i_{k(n-k)} \in \underline i$. For simplicity we will sometimes write $\bar x$ for $[x_{1\dots k}:\ldots : x_{n-k+1\dots n}]$ in the following.

\begin{algorithm}
\caption{}\label{alg:1}
{\normalsize
Input: $\mathcal{R}$, $e$
\begin{enumerate}
\item Find the equations defining $B_{2e}(\mathcal{R})$ in the Pl\"ucker coordinates, like explained in Section \ref{sec22}.
\item Solve the system of equations, that arise from $\bar x \bar H^{p}=0$, together with the equation of $B_{2e}(\mathcal{R})$, the shuffle relations and the equation $x_{1,\dots, k}=1$.
\end{enumerate}
Output: The solutions $\bar x = [x_{1\dots k}:\ldots : x_{n-k+1\dots n}]$ of this system of equations.
}
\end{algorithm}

\begin{thm}
 Algorithm \ref{alg:1} outputs the complete list $L$ of codewords (in Pl\"ucker coordinate representation), such that for each element $\bar x \in L$, $d_S(\varphi^{-1}(\bar x),\mathcal{R})\leq 2e$.
\end{thm}

\begin{proof}
The solution set to the shuffle relations is exactly $\varphi(\Gr)$, i.e. all the elements of $\mathbb{P}^{\binom{n}{k}-1}$ that are Pl\"ucker coordinates of a $k$-dimensional vector space in $\F_{q}^{n}$. The subset of this set with the condition $x_{1,\dots, k}=1$ is exactly the set of Pl\"ucker coordinates of elements in $\G$ whose reduced row echelon form has $I_{k}$ as the left-most columns. Intersecting this with the solution set of the equations given by $H^{p}$ achieves the Pl\"ucker coordinates of the lifted code $\Cvs$.
The intersection with $B_{2e}(\mathcal{R})$ is then given by the additional equations from 1. in the algorithm. Thus the solution set to the whole system of equation is the Pl\"ucker equations of $\Cvs \cap B_{2e}(\mathcal{R})$.
\qed \end{proof}

For the analysis of complexity of this algorithm we need to calculate the  number of equations, denoted by $\tau$, that define  a ball of radius $2e$.
\begin{lem}
 The number of equations defining $B_{2e}(\mathcal{U}_0)$ is equal to the number of equations defining $B_{2e}(\Uvs)$ for any $\Uvs \in \Gr$.
\end{lem}
\begin{proof}
 Follows directly from Lemma \ref{lem:2}.
\qed \end{proof}

Since we can count the elements that are not less than or equal to a given element in the Bruhat order, we get:

\begin{lem}
\label{lem:marko}
 The number of equations  defining $B_{2e}(\mathcal{U})$ inside $\Gr$ is
\[\tau = \sum_{l=0}^{k-e-1} \binom{n-k}{k-l} \binom{k}{l} = \binom{n}{k} - \sum_{l=k-e}^{k} \binom{n-k}{k-l} \binom{k}{l}.\]
\end{lem}
\begin{proof}
The condition that $(i_{1},\dots, i_{k}) \not \leq (e+1,\dots,k,n-t+1,\dots,n)$ is equivalent to
\[\exists l \in \{1,\dots,k-e\} : i_{l} >k .\]
For such an $l$ there are $k-l+1$ entries chosen freely from $\{k+1,\dots,n\}$ and $l-1$ entries from $\{1,\dots,k\}$. Hence there are
\[\sum_{l=1}^{k-e} \binom{n-k}{k-l+1} \binom{k}{l-1} = \sum_{l=0}^{k-e-1} \binom{n-k}{k-l} \binom{k}{l}\]
many elements in $\binom{[n]}{k}$ that are $\not \leq (e+1,\dots,k,n-t+1,\dots,n)$, which is equal to the number of equations defining $B_{2e}(\mathcal{U})$.
\qed \end{proof}

The complexity of Algorithm \ref{alg:1} is dominated by solving the system of $\tau+1+(\delta-1)(n-k)+\binom{n}{2k}$ linear and bilinear equations in $\binom{n}{k}$ variables. This has a complexity that is polynomial in $n$ and exponential in $k$.

In most of the examples we computed though, we only needed a subset of all equations to get the solutions. For this
note, that the actual information is encoded in the rank-metric code part of the matrix representation of the vector space, i.e. in the Pl\"ucker coordinates corresponding to $C^{p}$. Hence, one does not need the $k\times n$-matrix representation of the solutions from an application point of view, since the information can be extracted directly from the Pl\"ucker coordinate representation of the vector spaces. On the other hand, because of this structure it is also straight-forward to construct the matrix representation by using Lemma \ref{lem7} (i.e. without any computation needed). So, the number of variables in the system could be reduced to $k(n-k)$, and this can decrease the complexity of the algorithm.

\begin{ex}
We consider the code from Example \ref{ex8}.
\begin{enumerate}
  \item
Assume we received
\[\Rvs_1 = \left( \begin{array}{cccc} 1&0&1&0 \\ 0&0&0&1 \end{array}\right) .\]
We would like to correct one error. Thus we first find the equations for the ball of subspace radius~$2$:
\[B_{2}(\Uvs_{0}) = \{\Vvs=\rs(V) \in \mathcal{G}_{2}(2,4)  \mid M_{3,4}(V) = 0  \}\]
We construct $A^{-1}_1$ according to Construction \ref{con4}
\[A^{-1} _1= \left( \begin{array}{cccc} 1&0&0&1 \\ 0&0&1&0 \\ 0&0&0&1 \\ 0&1&0&0\end{array}\right) \]
and compute the last column of $\bar{\varphi}(A^{-1}_1)$:
\[[1:0:0:1:0:0]^{T}.\]
Thus, we get that
\[B_{2}(\Rvs_1) = \{\Vvs=\rs(V) \in \mathcal{G}_{2}(2,4)  \mid M_{1,4}(V)+M_{2,3}(V) = 0  \}.\]
Then combining with the parity check equations from Example~\ref{ex8} we obtain the following system of linear equations to solve
\begin{align*}
x_{13}+x_{14}+x_{24}&=0\\
x_{14}+x_{23}&=0\\
x_{12}+x_{23}&=0\\
x_{12} &= 1
\end{align*}
where the first two equations arise from $\bar H^{p}$, the third from $B_{2}(\Rvs_1) $ and the last one is the always given one. This system has the two solutions $(1,1,1,1,0)$ and $(1,0,1,1,1)$ for $(x_{12},x_{13},x_{14},x_{23},x_{24})$. Since we used all the equations defining the ball in the system of equations, we know that the two codewords corresponding to these two solutions (i.e. the third and fourth in Example \ref{ex8}) are the ones with distance $2$ from the received space, and we do not have to solve $x_{34}$ at all. The corresponding codewords are
\[ \left(  \begin{array}{cccc}
          1 & 0 & 1 & 0 \\
          0 & 1 & 1 & 1 \\  \end{array} \right) ,
    \left(  \begin{array}{cccc}
          1 & 0 & 1 & 1 \\
          0 & 1 & 0 & 1 \\  \end{array} \right) .\]

  \item

  Now assume we received
\[\Rvs_2 = \left( \begin{array}{cccc} 1&0&0&1 \\ 0&1&1&1 \end{array}\right) .\]
As previously, we construct $A^{-1}_2$ according to Construction \ref{con4}
\[A^{-1} _2= \left( \begin{array}{cccc} 1&0&0&1 \\ 0&1&1&1 \\ 0&0&1&0 \\ 0&0&0&1\end{array}\right) \]
and compute the last column of $\bar{\varphi}(A^{-1}_2)$:
\[[1:1:0:1:1:1]^{T}.\]
Thus, we get that
\[B_{2}(\Rvs_1) = \{\Vvs=\rs(V) \in \mathcal{G}_{2}(2,4)  \mid M_{1,2}(V)+M_{1,3}(V)+M_{2,3}(V)+M_{2,4}(V)+M_{3,4}(V) = 0  \}.\]
Then combining with the parity check equations from Example~\ref{ex8} and the shuffle relation from Example~\ref{ex:shuffle} we obtain the following system of linear  and bilinear equations
\begin{align*}
x_{13}+x_{14}+x_{24}&=0\\
x_{14}+x_{23}&=0\\
x_{12}+x_{13}+x_{23}+x_{24}+x_{34}&=0\\
x_{12}x_{34}+x_{13}x_{24}+x_{14}x_{23}&=0\\
x_{12} &= 1
\end{align*}

We rewrite these equations in terms of variables $x_{13},x_{14},x_{23},x_{24}$ which correspond to a lifted Gabidulin code as follows.
\begin{align*}
x_{13}+x_{14}+x_{24}&=0\\
x_{14}+x_{23}&=0\\
x_{1,3}+x_{2,3}+x_{2,4}+x_{13}x_{24}+x_{14}x_{23}&=1\\
\end{align*}

This system has three solutions $(1,0,0,1)$, $(0,1,1,1)$, and $(1,1,1,0)$  for $(x_{13},x_{14},x_{23},x_{24})$.  The corresponding codewords are
\[\left(  \begin{array}{cccc}
          1 & 0 & 0 & 1 \\
          0 & 1 & 1 & 0 \\  \end{array} \right),
\left(  \begin{array}{cccc}
          1 & 0 & 1 & 1 \\
          0 & 1 & 0 & 1 \\  \end{array} \right),
 \left(  \begin{array}{cccc}
          1 & 0 & 1 & 0 \\
          0 & 1 & 1 & 1 \\  \end{array} \right) .\]
         \end{enumerate}
\end{ex}

\begin{rem} Note that an upper and a lower bounds for the list size, i.e. the number of codewords in a ball of  subspace radius $2e$ around a received  word, can be directly derived from the bounds on a list size of a classical Gabidulin code, given rank radius $e$. This result follows from the next lemma.
\end{rem}

%
%

\begin{lem}\label{lem20}
Let $\Rvs \in \G$ and denote by $R\in \F_q^{k\times n}$ its reduced row echelon form. Then for any $A\in \F_q^{k\times (n-k)}$ there always exists a matrix $M\in\F_q^{k\times (n-k)}$ such that $d_S(\Rvs, \rs[\;I_k\; A\;])= d_S(\rs[\;I_k\; M\;], \rs[\;I_k\; A\;])$.
\end{lem}

\begin{proof}
Because of the reduced row echelon form it holds that there exists $\bar M\in \F_q^{k\times n-k}$ such that
\[\mathrm{rank}\left( \begin{array}{ccc}
I_k && A\\
&R&
                      \end{array}
 \right)=
\mathrm{rank}\left( \begin{array}{ccc}
I_k && A\\
0_{k\times k}&&\bar M
                      \end{array}
 \right)\]
which implies that $d_S(\Rvs, \rs[\;I_k\; A\;])= d_S(\rs[\;I_k\; A+\bar M\;], \rs[\;I_k\; A\;])$. With $M:=A+\bar M$, the statement follows.
\qed \end{proof}

Bounds for the list size for classical Gabidulin list decoding can be found e.g. in~\cite{WZ12}.

\section{Conclusion and Open Problems}
\label{sec:conclusions}

We presented a list decoding algorithm for lifted Gabidulin codes that works by solving a system of linear and bilinear equations in the Pl\"ucker coordinates. In contrast to the algorithms presented in \cite{GuXi12,MaVa12} this algorithm works for lifted Gabidulin codes for any set of parameters $q,n,k,\delta$.

One can easily extend the algorithm presented in this paper to work also for received spaces of a different dimension. For this, one only needs to change the conditions in Proposition \ref{prop5} indicating which Pl\"ucker coordinates have to be zero. The rest of the theory can then be carried over straight-forwardly. In a similar manner one can make the algorithm work for unions of LG codes of different length (cf. e.g. \cite{Ska10}). To do so, one needs to add a preliminary step in the algorithm where a rank argument decides, which of these LG codes can possibly have codewords that are in the ball around the received word.

The storage needed for our algorithm is fairly little, the complexity is polynomial in $n$ but exponential in $k$. Since in applications, $k$ is quite small while $n$ tends to get large, this is still reasonable. In future work, we still want to improve this complexity by trying to decrease the size of the system of equations to solve in the last step of the algorithm. Moreover, it would be interesting to see if some converse version of Theorem \ref{thm9} exists, i.e. if one can generate constant dimension codes from a given linear block code by using this as a subset of the Pl\"ucker coordinates of the constant dimension code. Moreover, we would like to find other families of codes that can be described through equations in their Pl\"ucker coordinates and use this fact to come up with list decoding algorithms of these other codes.



\end{document}